\documentclass{article}

\usepackage{wrapfig}
\usepackage{amsthm, amsmath,amssymb}
\usepackage[colorlinks,linkcolor=blue,urlcolor=blue,citecolor=black,plainpages=false,pdfpagelabels,breaklinks]{hyperref}
\usepackage[all]{xy}
\usepackage{graphicx}
\usepackage[margin=2.3 cm]{geometry}
\usepackage{caption}

\newtheorem{theorem}{Theorem}[section]
\newtheorem{definition}[theorem]{Definition}
\newtheorem{corollary}[theorem]{Corollary}

\newtheorem{remark}[theorem]{Remark}

\DeclareMathOperator{\EA}{EA}
\begin{document}

\title{{\sc Equivalence Relations in Quantum Theory:\\An Objective Account of Bases and Factorizations}\footnote{This paper is dedicated to the memory of Newton da Costa (1929-2024).}.}

\author{{\sc Christian de Ronde}$^{1,2,3}$, {\sc Raimundo Fern\'andez Mouj\'an}$^{2,4}$, {\sc C\'esar Massri}$^{5,6}$}
\date{}

\maketitle
\begin{center}
\begin{small}
1. Philosophy Institute Dr. A. Korn, University of Buenos Aires - CONICET\\
2. Center Leo Apostel for Interdisciplinary Studies\\Foundations of the Exact Sciences - Vrije Universiteit Brussel\\
3. Institute of Engineering - National University Arturo Jauretche\\
4. Philosophy Institute, Diego Portales University, Santiago de Chile\\
5. Institute of Mathematical Investigations Luis A. Santal\'o, University of Buenos Aires - CONICET\\
6. CAECE University
\end{small}
\end{center}

\begin{abstract}
\noindent In orthodox Standard Quantum Mechanics (SQM) bases and factorizations are considered to define {\it quantum states} and {\it entanglement} in relativistic terms. While the choice of a basis (interpreted as a measurement context) defines a state incompatible to that same state in a different basis, the choice of a factorization (interpreted as the separability of systems into sub-systems) determines wether the same state is entangled or non-entangled. Of course, this perspectival relativism with respect to reference frames and factorizations precludes not only the widespread reference to quantum particles but more generally the possibility of any rational objective account of a state of affairs in general. In turn, this impossibility ends up justifying the instrumentalist (anti-realist) approach that contemporary quantum physics has followed since the establishment of SQM during the 1930s. In contraposition, in this work, taking as a standpoint the logos categorical approach to QM ---basically, Heisenberg's matrix formulation without Dirac's projection postulate--- we provide an invariant account of bases and factorizations which allows us to  to build a conceptual-operational bridge between the mathematical formalism and quantum phenomena. In this context we are able to address the set of equivalence relations which allows us to determine what is actually {\it the same} in {\it different bases} and {\it factorizations}.  
\end{abstract}
\begin{small}

\end{small}

\bigskip

\bigskip

\bigskip

\bigskip

\bigskip

\bigskip

\bigskip

\section*{Introduction}

Since its formulation in the 1930s physicists have become accustomed to the idea that the reference to a {\it state} in Standard Quantum Mechanics (SQM) requires, to be meaningful, the specification of a particular experimental context or ``preferred'' basis (given by a complete set of commuting observables). This widely accepted perspectival relativization of the notion of state of a physical system has immense consequences, for it precludes the possibility in quantum theory to refer to a state of affairs as independent of the particular choice of a reference frame (a condition that one would believe is necessary for any objective physical representation). In contradistinction, factorizations in SQM have been misleadingly conceived in atomistic substantivalist terms, as informing the separation of quantum systems into subsystems ---disregarding the fact that this entails an operation provenly inconsistent with the mathematical formalism itself. Let us add that this understanding of factorizations linked to the notion of entanglement has resulted also in a problematic relativistic view according to which the choice of a factorization determines wether a state is entangled or not (see for a detailed discussion \cite{deRondeMassri23}).

But are these sacrifices necessary? Is that renunciation to the conditions of objectivity in physics in fact unavoidable? We would like to show that both the account of bases and of factorizations mentioned above have been generated due to the projection of presupposed classical concepts that have wrongly determined, right from the start, the understanding of the formalism and of quantum phenomena. However, by recognizing and escaping these inadequate projections, it is possible to restore an account of QM that is operationally invariant, that allows to account for a state of affairs independent of reference frames, and that opens the door for the development of concepts specifically adequate to this mathematical formalism \cite{deRondeMassri19a, deRondeMassri21a}. This conceptual account must be grounded on the \emph{intensive} nature of the elements pointed (already in Heisenberg's matrix mechanics) as invariant by the formalism itself. In this work we attempt, from a formal and conceptual reframing of the theory, to provide the operational conditions that, bridging the gap between the mathematical formalism and the conceptual representation of experience, will allow us to understand what is actually observed according to the theory of quanta. In short, the main attempt of this article is to provide a consistent and coherent formal, conceptual and operational redefinition of the notions of {\it basis} and {\it factorization} in invariant and objective terms.

\section{The Relativist Account of Bases}

A certain relativism seems to be implied in almost all accounts of QM. A relativism we believe was introduced in QM by no other than Niels Bohr. This was done through the redefinition of the notion of {\it measurement} within the theory of quanta in terms of the interaction between quantum and classical systems. A new interaction ---that implied the experimental choice of the observer--- which was made a fundamental aspect of the theory itself. Of course, as it is well known, measurement had never been part of any  physical theory before Bohr. Measurement in physics had always been understood as part of an epistemic praxis which presupposed a human subject capable to produce a link between the knowledge of an already ``closed'' theory and the highly complex technical experience produced in the lab. As Einstein would tell Heisenberg \cite{Heis71}: ``it is only the theory which decides what can be observed.'' The {\it moment of unity} of a theory, that which the theory talks about, is ---for the realist--- something to be defined through invariance and objectivity {\it within} the theory itself. It is in this way that classical mechanics generates the formal-conceptual notion of `particle' and electromagnetism is able to come up with the notion of `electromagnetic wave'. In short, we do not observe a unity given to us in experience ---we do not observe an electromagnetic wave--- but a multiplicity of phenomena ---a lightning or two rocks attracting each other--- that need to be united through mathematical formalisms and concepts. As Heisenberg would explain: 
\begin{quote}
\noindent {\small ``The history of physics is not only a sequence of experimental discoveries and observations, followed by their mathematical description; it is also a history of concepts. For an understanding of the phenomena the first condition is the introduction of adequate concepts. Only with the help of correct concepts can we really know what has been observed.''   \cite[p. 264]{Heis73}}
\end{quote}
But, instead of introducing the concepts specifically adequate to the quantum formalism ---and capable of producing an understanding of the observations coherent with that formalism--- Bohr would impose not only a reference to `waves' and `particles' that defined a necessary restriction to classical concepts \cite{Howard94}, but also ---since quantum phenomena escaped a consistent representation in terms of particles or waves--- a contextual dependency of this reference to the specific experimental setup. According to Bohr, that reference ---as repeatedly addressed within the double slit experiment--- was explicitly dependent on the choice of the specific measurement situation (see for a detailed analysis \cite{deRonde20}). Contrary to all pre-Bohrian physics, the nature of the state of affairs would then come to depend on the choice of the experimental context, and consequently, each different measurement would be understood as referring to a different (context-dependent) state of affairs. The complementary nature of quantum entities would then preclude a common consistent reference to {\it the same} state of affairs. In short, Bohr's new understanding of measurement would define for each measurement context a different ``classical'' state of affairs, establishing a relativist view of quantum phenomena. And it is exactly this same relativist operation that would be reproduced by Dirac in his redefinition of the notion of (quantum) {\it state} in terms of a preferred reference frame (or basis). In his famous book \cite{Dirac74}, he would extend complementarity from Bohr's conceptual dualistic reference to `waves' and `particles' to the formal level of vectors and bases.\footnote{It is important to notice that the notion of `observable' was introduced by Dirac to refer to linear operators (in the first edition from 1930) and to real functions of dynamical variables (in the second edition from 1935).} In this way, the English engineer and mathematician would drastically subvert the meaning of {\it state of a system} by imposing a direct link between vectors in specific bases ---namely, those in which a vector is written as a single ket, $| x \rangle$--- and actual observations of single `clicks' in detectors. It is this re-interpretation of the notion of {\it state} in terms of a measurement outcome, mathematically represented by a {\it ket} vector, which would impose the ---today widely accepted--- idea that {\it the same system} could be now represented in terms of {\it different states} depending on the specific basis (i.e., the reference frame).
\begin{quotation}
\noindent {\small ``[E]ach state of a dynamical system at a particular time corresponds to a ket vector, the correspondence being such that if a state results from the superposition of certain other states, its corresponding ket vector is expressible linearity in terms of the corresponding ket vectors of the other states, and conversely. Thus the state $R$ results from a superposition of the states $A$ and $B$ when the corresponding ket vectors are connected by $ | R \rangle = c_1 | A \rangle + c_2 | B \rangle $.'' \cite[p. 16]{Dirac74}}
\end{quotation}
This is an essentially inconsistent account of the notion of state which not only precludes its basic meaning as referring to something independent of any basis, but also mixes the basis dependent definition of state with the abstract reference to vectors. The state $ | R \rangle$ is {\it different} to the state $ | A \rangle$ and {\it different} to the state $ | B \rangle$, which means that when considering  different bases  or experimental situations we obtain different states, different outcomes. However, the state  $ | R \rangle$ is also considered as being {\it the same} as the sum of the states $ | A \rangle$ and $ | B \rangle$. Thus, through the identification of states with singular outcomes (taken to be expression of specific particles), each of the states is completely `different' from the others, and yet they are also `the same'. How is this paradox created? As discussed in detail in \cite{deRondeMassri22a}, the contradiction is built trough the implicit reference to two different (inconsistent) definitions of the notion of state. The first basis-dependent definition is in terms of a single term {\it ket} which, in turn, is also linked to the observation of a single outcome ---i.e., each state  $| k \rangle$ generates a specific measurement outcome `k'.
\begin{definition}[Operational Purity] Given a quantum system in the state $|\psi \rangle$, there exists an experimental situation linked to that basis (in which the vector is written as a single term) in which the test of it will yield with certainty (probability = 1) its related outcome. 
\end{definition}
\noindent In this case each different {\it ket} (e.g., $ | R \rangle$, $ | A \rangle$ or $ |B \rangle$) is interpreted as a different {\it state}. However, there is also an implicit reference to a completely different notion of state in terms of the `abstract vector' which grounds these different basis dependent representations. But notice that this later definition is independent of bases. In this case, $ | R \rangle$ and $c_1 | A \rangle + c_2 | B \rangle $ are just particular basis-dependent representations of {\it the same} abstract vector $\Psi$. 
\begin{definition}[Abstract Purity]\label{pure} An abstract unit vector (with no reference to any basis) in Hilbert space, $\Psi$, is a pure state. In terms of density operators $\rho$ is a pure state if it is a projector, namely, if Tr$(\rho^2) = 1$ or $\rho = \rho^2$. 
\end{definition}
\noindent As demonstrated explicitly in \cite{deRondeMassri22a}, while the first definition is operational but non-invariant, the latter mathematically abstract definition is invariant but has no operational content. Furthermore, they are not equivalent, while the first definition implies the second the inverse does not hold (for a more detailed analysis see \cite{deRondeMassri22b}).

It is interesting to notice that it is this self-contradictory relativist account of state proposed by Dirac which would be applied by Bohr a few years later in his famous reply to EPR. As Bohr \cite[p. 696]{Bohr35} would argue, the condition to consider a subset of observables as definite valued required the specification of the reference frame (i.e., the measurement context or basis): ``it is a well-known feature of the present formalism of quantum mechanics that it is never possible, in the description of the state of a mechanical system, to attach definite values to both of two canonically conjugate variables''. In this way, complementarity would become extended in terms of the contextual reference to states and bases.
\begin{quotation}
\noindent {\small ``the renunciation in each experimental arrangement of the one or the other of two aspects of the description of physical phenomena, ---the combination of which characterizes the method of classical physics (...)--- depends essentially on the impossibility, in the field of quantum theory, of accurately controlling the reaction of the object on the measuring instruments, i.e., the transfer of momentum in case of position measurements, and the displacement in case of momentum measurements.''  \cite[p. 699]{Bohr35}} 
\end{quotation}   

Ever since, it has been dogmatically accepted that the reference to a {\it state} in QM requires the specification of a basis. Different bases can be compared but only in a complementary (i.e., inconsistent) fashion. Thus, quantum properties (i.e., projection operators) cannot be assigned a {\it global valuation} and the description of a state of affairs becomes intrinsically {\it relative} to the single (preferred) context (or basis) in which the measurement is actually performed. This has led many to argue that quantum states must be understood as essentially relative \cite{Rovelli96} or perspectival \cite{Dieks22} notions (see for a detailed discussion and analysis \cite{deRondeFM18}). As a consequence, it is commonly argued that ``the properties of a system are different whether you look at them or not'' \cite{Butterfield17}.

%W. H. Zurek: Pointer basis of quantum apparatus: Into what mixture does the wave packet collapse? Phys. Rev. D24 (1981), 1516?1525.

%A. O. Barvinsky and A. Y. Kamenshchik: Preferred basis in quantum theory and the problem of classicalization of the quantum universe. Phys. Rev. D 52 (1995), 743?757.

\section{The Substantialist Account of Factorizations}

The orthodox account of QM established during the 1930s is grounded on Bohr's atomistic interpretation of the theory according to which there is a microscopic realm composed of elementary particles which, even though are actually responsible for the very existence of `tables' and `chairs' cannot be represented because of the ``uncontrollable disturbance'' between quantum and classical systems ---due to the {\it quantum of action}--- within the measurement interaction. Dirac's work would reinforce this substantialist reference to particles within his intrinsically instrumentalist development of the theory. In fact, the unjustified reference to particles would be used by Dirac in order to justify the shift from the observation of {\it intensive patterns} ---which Heisenberg's had considered within his matrix QM--- to the unilateral attention to single `clicks' in detectors regarded ---following Bohr's ``commonsensical'' reasoning--- as the obvious consequence of (presupposed) ``microscopic particles''. This presupposition also resulted in Dirac's {\it ad hoc} introduction of the ``collapse'' process, imposed in order to bridge the gap between the intensive values present within quantum superpositions and the observation of the single measurement outcome ---that had to be taken as the main experience to explain, since it expressed the presence of that already presupposed `particle'.  
\begin{quotation}
\noindent {\small ``When we make the photon meet a tourmaline crystal, we are subjecting it to an observation. We are observing whether it is polarized parallel or perpendicular to the optic axis. The effect of making this observation is to force the photon entirely into the state of parallel or entirely into the state of perpendicular polarization. It has to make a sudden jump from being partly in each of these two states to being entirely in one or other of them.'' \cite[p. 7]{Dirac74}}
\end{quotation}   
It is at this point, that we need to recognize that the unjustified reference to particles was never ``just a way of talking'' \cite{deRondeFM21} but ---on the very contrary--- implied from the beginning a series of essential methodological operations that would even determine the formalism of the theory itself. As Faraday explained long ago: “the word \emph{atom}, which can never be used without involving much that is purely hypothetical, is often intended to be used to express a simple fact; but good as the intention is, I have not yet found a mind that did habitually separate it from its accompanying temptations” \cite[p. 220]{Laudan81}. Schrödinger rephrases this idea for the quantum case: “We have taken over from previous theory the idea of a particle and all the technical language concerning it. This idea is inadequate. It constantly drives our mind to ask information which has obviously no significance” \cite[p. 188]{Schr50}. It is not difficult to understand that if one  dogmatically applies a series of categorical principles ---such as those of particle metaphysics, e.g., separability, individuality, locality, etc.--- to a mathematical formalism that was never meant to be understood under the constraints of such representation, the result of this methodology can only lead to paradoxes and dead ends. An excellent example of how pseudo-problems can be built by following this methodology ---where dogma is imposed over reason and experience--- is the substantialist account of factorizations, that in SQM are interpreted as the ``separability of systems into sub-systems'', an interpretation put forward by the Russian mathematician and physicist Lev Landau during the late 1920s. Landau's account of factorizations imposed the dogmatic application of the principle of separability derived from the metaphysics of particles inherent to classical physics. In his work from 1927 the projection of subspaces was incorrectly interpreted in terms of the ``separation of a system into sub-systems''. Importing the classical way a reasoning about classical systems to QM, Landau would write in \cite{Landau27} about ``coupled systems'' and the necessity of a ``probability ensemble''. According to Landau \cite[p. 8]{Landau27} : ``A system cannot be uniquely defined in wave mechanics; we always have a probability ensemble (statistical treatment). If the system is coupled with another, there is a double uncertainty in its behaviour.'' He pointed out that if a system is coupled with another, then the uncertainty grows. If the first system can be described by quantities $\{a_i\}$, 
\[
|x\rangle = \sum_i a_i|x_i\rangle
\]
and the second system by $\{b_j\}$, 
\[
|x'\rangle = \sum_j b_j|x'_j\rangle,
\]
then the two systems together are described by $\{c_{ij}=a_ib_j\}$,
\[
|xx'\rangle = \sum_{ij} a_ib_j |x_ix_j'\rangle.
\]
If the two systems are coupled, the coefficients $c_{ij}$ depend on time and can no longer be resolved as a product. Now, let us see why this interpretation is completely meaningless within the formalism of QM.  

The notion of separability is grounded on the modern metaphysical representation provided by classical physics according to which physical reality is composed of independent separated individual entities which exist within space and time. According to this supposedly ``commonsensical'' picture, a system can be understood in terms of its parts and the knowledge of these parts implies the knowledge of the whole individual. This is of course a direct consequence of the underlying Boolean logic that is a prerequisite to define the notion of {\it actual entity} in terms of the principles of existence, non-contradiction and identity ---also grounding classical mechanics. Indeed, as a consequence, the propositions derived from classical mechanics can be arranged in a Boolean lattice (see for a detailed discussion \cite{deRondeFreytesDomenech18}). According to classical logic, and following set theory, the {\it sum} or {\it union} of the elements of a system imply its complete characterization as a whole. It is easy to picture all of this in terms of sets and the logical relations that we learned in primary school: 
\begin{center}
\includegraphics[scale=.35]{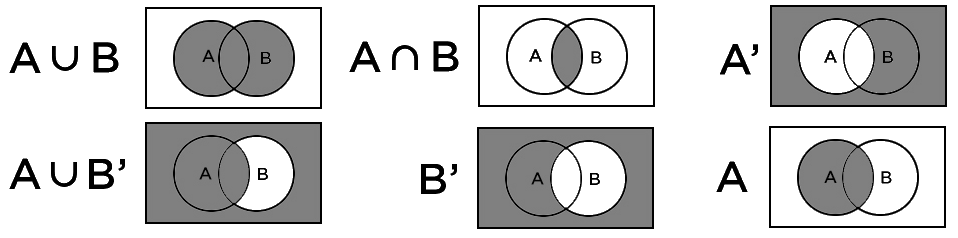}
\captionof{figure}{{\it Union}, {\it intersection} and {\it complement} in Boolean logic.}
 \end{center}
\noindent However, as it is also well known since the famous paper by Birkhoof and von Neumann \cite{BvN36}, the underlying logic of QM is not Boolean, it is not {\it distributive}. Thus, the basic classical way of reasoning about systems becomes precluded right from the start. This is an obvious consequence of the fact that vectorial spaces do not relate between each other following the same rules as the elements of a set through {\it union} and {\it conjunction}. In the quantum case the equivalent to the {\it union} of two vectors is not the {\it sum} of the individual vectors considered as lines, but instead what they are capable to {\it generate} in terms of subspaces. Thus, in the particular case when we consider the {\it sum} of two vectors what we obtain is the whole plane.
\begin{center}
\begin{tabular}{ccc}
\includegraphics[scale=.45]{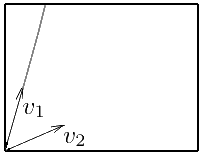} &
\includegraphics[scale=.45]{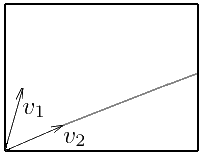} &
\includegraphics[scale=.45]{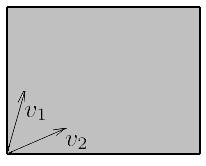} \\
$\langle v_1\rangle$  &
$\langle v_2\rangle$ &
$\langle v_1\rangle\oplus\langle v_2\rangle$
\end{tabular}
\end{center}

As it could have been easily foreseen, the artificial {\it ad hoc} introduction of a set of logical relations completely alien and even incompatible with the mathematical formalism of the theory could only lead to confusions, contradictions and pseudo-problems. Sadly enough, this is exactly what happened with the introduction of the notion of separability in the context of QM. As we just explained, the \emph{union} of two vectors was inadequately understood as a \emph{sum} when, in fact, it is a \emph{generation}. Analogously, the \emph{projection} of a subspace was incorrectly interpreted as a {\it separation} of the whole set and the choice of a subset of elements, when, as a matter of fact, its correct interpretation is that of {\it shadow} \cite{deRondeMassri23}. In figure 2 we can see that while the shadow of $|\Psi\rangle$ in the $x$-axis is the vector $|x\rangle$, the shadow of $|\Psi\rangle$ in the $y$-axis is the vector $|y\rangle$\\ 
\begin{center}
\includegraphics[scale=.3]{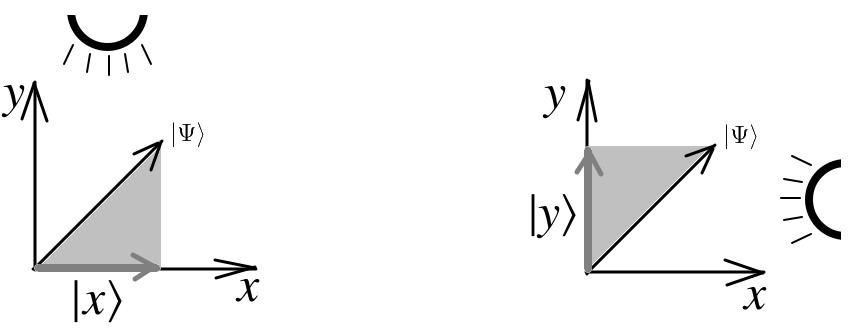} 
\captionof{figure}{The {\it shadow} of $|\Psi\rangle$ in the $x$-axis, $|x\rangle$, and in the $y$-axis, $|y\rangle$.}
\end{center}

The concept of ``subsystem'' and that of ``product state'' lead to the set-theoretic idea that the ``system'' is the union of its parts ---which, as it is very well known, is not the case in the quantum formalism and vectorial spaces where there are no `parts' (i.e., elements of a set) but `shadows' (i.e., projection of subspaces). To say that the projections of subspaces should be understood ---following the classical line of reasoning--- as ``subsystems'' in SQM leads then to essential contradictions. In the early 1980s Diederik Aerts would derive a non-separability theorem in the context of quantum logic which demonstrated that quantum systems are essentially non-separable \cite{Aerts81, Aerts84b}. Later on, in the mid-1990s Rob Clifton would show in \cite{Clifton95} the inconsistencies of imposing the notion of separability to the orthodox formulation. Taking into consideration the following states: 
\begin{align*}
|\psi\rangle & = \sqrt{w_1}|\text{Decay after time }t\rangle_{\text{Atom}}\otimes |\text{Dead}\rangle_{\text{Cat}} \\
&+ \sqrt{w_2}|\text{No decay after time }t\rangle_{\text{Atom}}\otimes |\text{Alive}\rangle_{\text{Cat}} 
\end{align*}
for $w_1+w_2=1$, $0<w_1<w_2<1$
he would show how it is easy to find a contradiction when attempting to consider the definite values of states. Clifton begins by making explicit two central claims to which any realist interpretation should subscribe:
\begin{itemize}
\item[{\bf MI1:}] Schr\"odinger's cat is definitely alive or dead despite its entanglement in $|\psi\rangle$.
\item[{\bf MI2:}] At any given time, only a proper subset of all possible properties of a given system corresponds to the definite properties it actually possesses.
\end{itemize}
Let $\text{Def}_{|\phi\rangle}(S)$ be the set of particular intrinsic properties a system $S$ possesses in the state
$|\phi\rangle$, let $\{P(S)\}$ be the set of all properties of $S$ and let $P_{|\phi\rangle}$ be
the projection operator with one dimensional range generated by the vector $|\phi\rangle$. Then, 
we can translate the assertions mathematically as
\begin{itemize}
\item[{\bf MI1:}] $\{P_{|\text{Alive}\rangle_{\text{Cat}}}, P_{|\text{Dead}\rangle_{\text{Cat}}}\}\subseteq \text{Def}_{|\psi\rangle}(\text{Cat})$
\item[{\bf MI2:}] $\text{Def}_{|\psi\rangle}(\text{Atom + Cat})\subseteq \{P(\text{Atom + Cat})\}$.
\end{itemize}
And using the standard formalism of QM we obtain the following, 
\begin{itemize}
\item[{\bf CP1:}] If $\text{Tr}(P|\psi\rangle)=1$ then $P\in \text{Def}_{|\psi\rangle}(\text{Atom + Cat})$.
\item[{\bf CP2:}] $\text{Def}_{|\psi\rangle}(\text{Atom + Cat})$ is an ortholattice.
\item[{\bf CP3:}] If $P\in\text{Def}_{|\psi\rangle}(\text{Cat})$ then $I\otimes P \in \text{Def}_{|\psi\rangle}(\text{Atom + Cat})$.
\end{itemize}
Clifton proves that  {\bf MI1} negates {\bf MI2}, hence both claims cannot coexist. Notice that {\bf CP3} formalizes the fact that the intrinsic properties of a composite system should at least include the intrinsic properties of its parts. Then he provides an example of the inconsistency leading to the lack of this necessary condition to account for systems and their parts in any meaningful manner: 
\begin{quotation}
\noindent {\small ``the intrinsic properties of a composite system should at least include the intrinsic properties of its parts. For example, whenever its true (respectively, false) that `the left-hand wing of the 747 has the property of being warped', then it must surely (classically, anyways!) also be true (respectively, false) that `The 747 has the property that its left-hand wing is warped'. It may seem that the difference between these two propositions is inconsequential; but in fact that we take the former to entail the latter makes all the difference to whether we are confident flying the 747!'' \cite{Clifton95}}
\end{quotation}

Maybe a simple image can expose why the attempt to understand shadows as subsystems is obviously misleading. 
\begin{center}
\includegraphics[scale=.15]{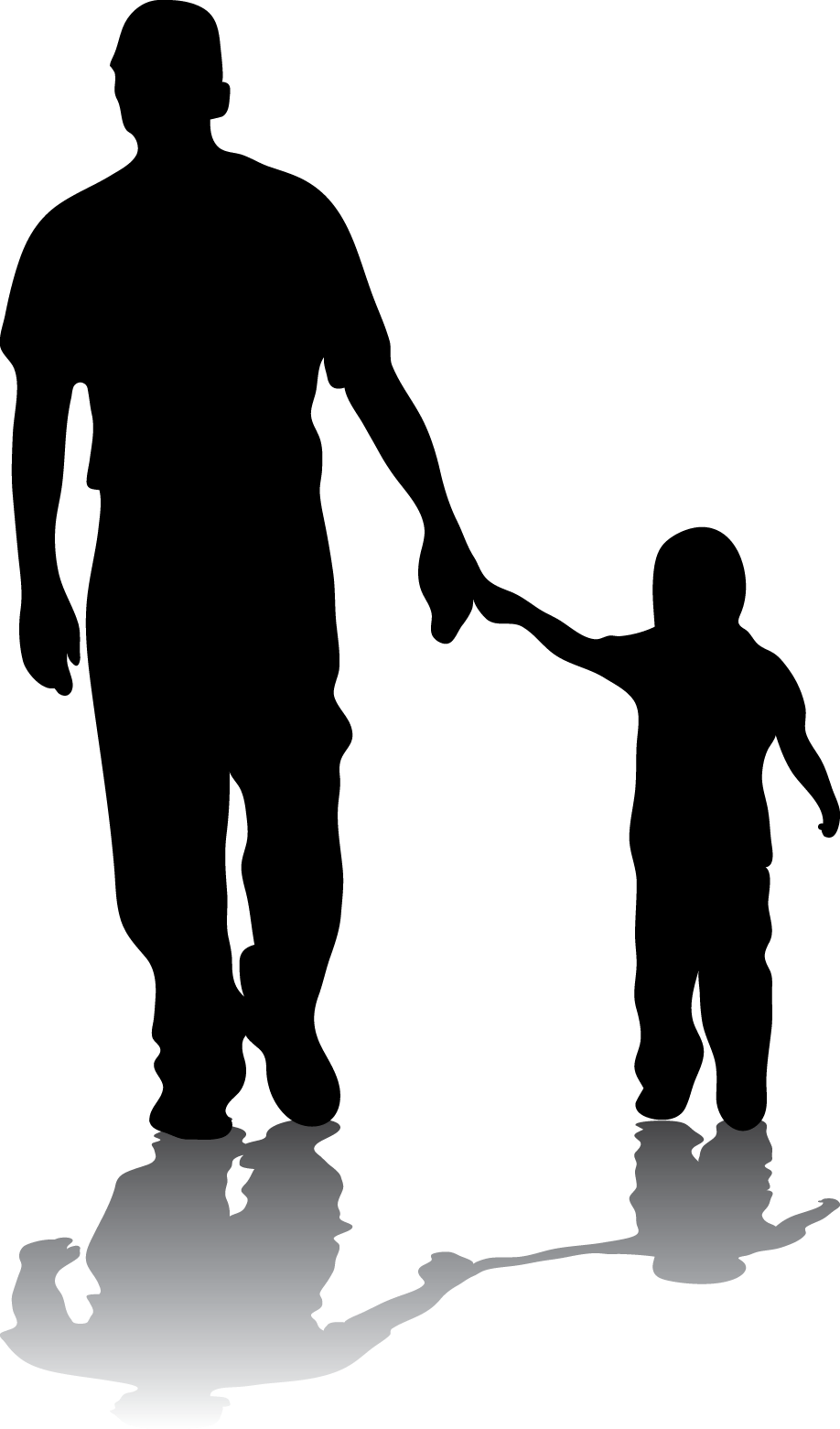}\\
{\small ``Dad, those are our subsystems, right?''}
 \end{center}

\section{The Logos' Realist Approach to QM: Operational-Invariance}

The relativization of the notion of physical state in SQM has precluded the understanding of physics as a discipline which is capable of producing an invariant-objective representation of a state of affairs independent of reference frames or empirical viewpoints, opening the door to the establishment of the 20th century instrumentalist account of physics. This is expressed in the widely accepted claim according to which, since we cannot understand what the theory talks about nor produce a detached representation of physical reality, we must content ourselves with a pragmatic approach to the quantum formalism understood as a ``recipe'' to predict measurement outcomes. As Bohr \cite{Bohr60} would state explicitly: ``Physics is to be regarded not so much as the study of something a priori given, but rather as the development of methods of ordering and surveying human experience.'' Today, this continues to be the orthodox contemporary understanding of the theory of quanta. As recently described by Tim Maudlin: 
\begin{quotation}
\noindent {\small ``What is presented in the average physics textbook, what students learn and researchers use, turns out not to be a precise physical theory at all. It is rather a very effective and accurate recipe for making certain sorts of predictions. What physics students learn is how to use the recipe. For all practical purposes, when designing microchips and predicting the outcomes of experiments, this ability suffices. But if a physics student happens to be unsatisfied with just learning these mathematical techniques for making predictions and asks instead what the theory claims about the physical world, she or he is likely to be met with a canonical response: Shut up and calculate!'' \cite[pp. 2-3]{Maudlin19}} 
\end{quotation} 

However, most of these claims regarding the impossibility of theoretical representation are just consequence of imposing a dogmatic atomist picture. In fact, in QM it is ---and always has been--- possible to determine theoretically a state of affairs independent of reference frames. Right from the beginning ---in Heisenberg’s matrix mechanics--- QM was constructed in invariant terms  allowing to determine, within its formalism, physical quantities that would remain the same for different reference frames. This invariant character of the formalism has been however generally ignored or casted aside, due to the projection of classical concepts, incompatible with the quantum formalism. Instead of developing the new concepts adequate to the invariant content of the formalism, the efforts were directed ---first by Sch\"odinger and later on by Dirac--- towards an (impossible) restoration of a “common sense” view of the physical world, which was classical, and which was believed to be the only reasonable one. This mandatory projection of an understanding incompatible with QM has made us leave aside that invariant aspect of the formalism. But let us explain this further. 

In principle, at least in formal terms, it is operational-invariance which defines what has to be considered as \emph{the same} in a physical theory. The essential role played by invariance consists in that it allows to consider experience from different perspectives in a consistent manner as referring to the same state of affairs. In turn, as Einstein constantly remarked, this makes possible to address the existence of a state of affairs as detached from particular empirical observations and reference frames. Given a mathematical formalism of a theory which has an invariant transformation, it does not really matter which particular reference frame we might choose to describe the state of affairs simply because there will be a consistent translation between any of them. We can go from one representation to another without loosing consistency and coherency regarding what is considered to be {\it the same} state.\footnote{Of course, thsi does not mean that the properties considered from different reference frames will possess the same values (e.g., position and velocity in classical mechanics), it means there will exist a consistent translation of these properties from one reference frame to the other.} This is surely realized in Newtonian mechanics, in Maxwell’s electromagnetism, and also —via the Lorentz’ transformations— in relativity theory. But is it possible to find invariance in QM? Contrary to what is usually believed, the answer is yes, and ---as we already mentioned above--- right from the start, in Heisenberg’s matrix mechanics. 

For some years, Heisenberg had followed Bohr’s guide, focusing on the question of describing the trajectories of electrons inside the atom. But the critical reaction of Wolfgang Pauli and Arnold Sommerfeld led him to think that he should take a different path \cite{BokulichBokulich20}. So, instead of trying to describe trajectories of unseen, presupposed, corpuscles, Heisenberg reframed the problem in terms of observable quantities. As explained by Jaan Hilgevoord and Joos Uffink \cite{HilgevoordUffink01}: “His leading idea was that only those quantities that are in principle observable should play a role in the theory, and that all attempts to form a picture of what goes on inside the atom should be avoided. In atomic physics the observational data were obtained from spectroscopy and associated with atomic transitions. Thus, Heisenberg was led to consider the ‘transition quantities’ as the basic ingredients of the theory.” That same year, he would present his groundbreaking results in the following manner \cite{Heis25}: “In this paper an attempt will be made to obtain bases for a quantum-theoretical mechanics based exclusively on relations between quantities observable in principle.” Emancipating himself completely from the atomist picture, Heisenberg was able to create a completely new mathematical formalism. As he would recall in his autobiography:
\begin{quotation}
\noindent {\small ``In the summer term of 1925, when I resumed my research work at the University of G\"ottingen ---since July 1924 I had been {\it Privatdozent} at that university--- I made a first attempt to guess what formulae would enable one to express the line intensities of the hydrogen spectrum, using more or less the same methods that had proved so fruitful in my work with Kramers in Copenhagen. This attempt lead me to a dead end ---I found myself in an impenetrable morass of complicated mathematical equations, with no way out. But the work helped to convince me of one thing: that one ought to ignore the problem of electron orbits inside the atom, and treat the frequencies and amplitudes associated with the line intensities as perfectly good substitutes. In any case, these magnitudes could be observed directly, and as my friend Otto had pointed out when expounding on Einstein's theory during our bicycle tour round Lake Walchensee, physicists must consider none but observable magnitudes when trying to solve the atomic puzzle.'' \cite[p. 60]{Heis71}}
\end{quotation}

Heisenberg was capable of developing matrix mechanics following two ideas: first, to leave behind the classical notion of particle-trajectory, as it did not seem required by QM —and it rather appeared as a classical habit that was inadvertently coloring a priori the approach to the understanding of the new theory—, and second, to take as a methodological standpoint Ernst Mach’s positivist idea according to which a theory should only make reference to what is actually observed in the lab. And what is actually observed is well known to any experimentalist: a spectrum of line intensities. This is what was described by the tables of data that Heisenberg attempted to mathematically model and that finally led him —with the help of Max Born and Pascual Jordan— to the development of the first mathematical formulation of the theory of quanta. Let us stop to take note once again of some of the conditions that were fundamental for the development of the quantum formalism. First, Heisenberg’s abandonment of Bohr’s atomist narrative and research program which focused in the description of unobservable trajectories of presupposed yet irrepresentable quantum particles. Second, the consideration of Mach’s observability principle as a methodological standpoint that ---even if Heisenberg didn't fully embraced the positivist credo--- allowed him to find a starting point unburdened of those classical presuppositions. That methodological standpoint led him finally to the replacement of Bohr’s fictional trajectories of irrepresentable electrons by the consideration of the intensive quantities appearing in line spectra that were actually observed in the lab. And these quantities, once detached from a supposedly necessary reduction to atomic elements, were what the formalism was indicating as invariant. Radically new, and of fundamental importance to produce a consistent and invariant quantum formalism, was this idea that we should accept intensive values as basic, as perfectly good “substitutes”, that are in no need whatsoever to be reconduced to binary values. Intensities appeared as basic and sufficient. But Heisenberg's intuition, according to which we should take as mainly significant the intensive patterns, was mostly discarded. As we explained earlier, the replacement of intensive patterns by single outcomes as what was most meaningful took place especially in Dirac's work, and as a consequence of a dogmatic presupposition: we must refer to particles, and those single outcomes must be the expression of specific particles.

Perhaps it may also be of interest to consider an unexpected conclusion emerging from what we have just discussed. The orthodox understanding of current physics is largely justified on an empiricist foundation, wherein the theoretical work of physics must always start from observations, considered as uncontaminated data, stripped of any theoretical presupposition. And it is on this foundation, on the supposedly theoretically innocent nature of observations, that the objectivity of physics is somehow justified. Beyond the problems behind this idea of physics (especially in its conception of observation), the truth is that what has been said so far seems to suggest that this condition may not have been met, as commonly believed. One could argue that the only one who remained firmly rooted in what is actually observed, the only radical empiricist, the only one who managed to shed inadequate assumptions and consider anew what was actually observed (and, from that, develop the formalism of quantum mechanics), was Werner Heisenberg. He managed to free himself from the atomistic narrative that imposed a series of dogmatic presuppositions and could thus confront what was observed directly, concluding in the invariance of intensities suggested by his matrix mechanics. From all of this, it becomes apparent that what happened afterward was, in large part, a gradual departure from the radical attention to what is actually observed. By reinstating the atomistic dogma and reintroducing the fundamentals of classical concepts, invariance was destroyed, and a series of presuppositions were imposed on observations that were not only foreign to them but, even worse, were incompatible with them. The history of quantum mechanics seems to describe, after Heisenberg's matrix mechanics, an increasing departure from the focus on what is observed. Less empiricism, not more ---at least if we decide to take it seriously.

In any case, a physical theory is not only an invariant mathematical formalism. Equally fundamental is the development of a conceptual representation that allows us to qualitatively understand what is physically real according to the formalism, and to give meaning to the observations that are predicted by the formalism. Together with formal invariance, conceptual objectivity is a fundamental aspect of any physical theory. But this conceptual representation cannot be an arbitrary addition. An objective conceptual representation has to be developed in strict accordance with the conditions established by the formalism. Specifically, it should be grounded on the {\it moment of unity} produced through the invariant content of the formalism. Constructing, in turn, an objective conceptual representation of what appears as invariant, producing the physical concepts that correspond to the invariant aspects of the formalism. If, as we said, what appears as invariant are the intensive quantities, and if this invariance is lost when we attempt to redirect intensities to binary values, we should start by producing the concept of an originally intensive element of physical reality, which is sufficient, which does not entail the redirection to other elements that could be understood in a binary manner. In this respect, we propose the concept of \emph{power of action}, or \emph{intensive power}, which represents those invariant intensive quantities present in the quantum formalism (without the need of adding any projection postulate). The intensive ---or probabilistic--- value of each power is termed its \emph{intensity} or \emph{potentia}. If we think about it, in fact the reference to a physical reality of an intensive nature is already present in the way Max Planck formulated his original discovery: the \emph{quantum of action}. Action surely is no particle; action represents a reality perhaps more intuitively conceived in intensive terms. The reference is also found in configuration space which, in fact, is nothing but a space which encapsulates degrees of action. In this representation, QM talks about ``action'', about powers of action, quantitatively defined by their intensity or potentia. This representation entails a rejection of the atomistic image of the world, that forces us constantly to take the probabilistic values as a diminished, insufficient representation, as evidently not the ‘real thing’, and it also entails that it is not necessary to redirect those intensive quantities in each case to binary values in order to determine a state of affairs. Let us add that the notion of a physical element that is in itself intensive implies a different understanding of observation. For instance, in the particular experimental situations where we obtain a single outcome at a time, it is only possible to obtain a measure of the physical element considered within the theory through repetition. This means that, in those situations, a single outcome is not what is mainly meaningful, but, on the contrary, an insufficient information, a minimal and partial measure. This is an example of how the understanding of observation is conditioned by theoretical presuppositions ---contrary to the contemporary widespread empiricist understanding according to which observations, taken as pure, uninterpreted data, define an uncontaminated starting point.

When we stick to the intensive values of action, we are able in fact to refer to a state of affairs that is independent of the particular representation in a reference frame (or basis), escaping thus the relativism with which most accounts of QM have contented themselves. In contraposition to an  {\it Actual State of Affairs} (ASA), defined classically in terms of a set of true definite valued binary properties, we propose to relate the reference of QM to an {\it Intensive State of Affairs} (ISA) ---also called elsewhere  {\it Potential State of Affairs} (PSA) \cite{deRondeMassri21a}. What has been demonstrated is that by considering an {\it intensive}, rather than {\it binary}, state of affairs, it is possible to restore a consistent global valuation for all projection operators independently of the basis. Let us recall some results from \cite{deRondeMassri21a}. While a \emph{Global Binary Valuation} (GBV) is a function from a graph to the set $\{0,1\}$, a \emph{Global Intensive Valuation} (GIV) is a function from a graph to the closed interval $[0,1]$. We term projection operators as {\it intensive powers}.\footnote{For a detailed introduction, analysis and discussion of the notion of `intensive power' we refer the interested reader to \cite{deRonde16}, and more specifically, \cite[Sect. 8]{deRondeMassri21a} and \cite[Sect. 3]{deRondeMassri19a}.} Let $H$ be a Hilbert space and let $\mathcal{G}=\mathcal{G}(H)$ be the set of observables. We give to $\mathcal{G}$ a graph structure by assigning an edge between observables $P$ and $Q$ if and only if $[P,Q]=0$. We call this graph, \emph{the graph of powers}. Among all global intensive valuations we are interested in the particular class of ISA.
\begin{definition}
Let $H$ be a Hilbert space. An \emph{Intensive State of Affairs} is a global intensive valuation
$\Psi:\mathcal{G}(H)\to[0,1]$ from the graph of powers $\mathcal{G}(H)$
such that $\Psi(I)=1$ and 
\[
\Psi(\sum_{i=1}^{\infty} P_i)=
\sum_{i=1}^\infty \Psi(P_i)\]
for any piecewise orthogonal projections $\{P_i\}_{i=1}^{\infty}$.
The numbers $\Psi(P) \in [0,1]$, are called {\it intensities} or {\it potentia}
and the nodes $P$ are called \emph{ powers}.
Hence, an ISA assigns a potentia to each power.
\end{definition}
Intuitively, we can picture an ISA
as a table,
\[
\Psi:\mathcal{G}(H)\rightarrow[0,1],\quad
\Psi:
\left\{
\begin{array}{rcl}
P_1 &\rightarrow &p_1\\
P_2 &\rightarrow &p_2\\
P_3 &\rightarrow &p_3\\
  &\vdots&
\end{array}
\right.
\]

\begin{theorem}
Let $H$ be a separable Hilbert space, $\dim(H)>2$ and let $\mathcal{G}$ be the graph of powers with the commuting relation given by QM.
\begin{itemize}
\item Any positive semi-definite self-adjoint operator 
of the trace class $\rho$ determines in a bijective way
an ISA $\Psi:\mathcal{G}\to [0,1]$. 
\item Any GIV determines univocally a set of powers that are considered as truly existent. 
\end{itemize}
\end{theorem}
\begin{proof} \begin{enumerate}
\item Using Born's rule, we can assign to each
observable $P\in\mathcal{G}$ the value $\mbox{Tr}(\rho P)\in[0,1]$.
Hence, we get an ISA $\Psi:\mathcal{G}\to[0,1]$.
Let us prove that this assignment is bijective. Let 
$\Psi:\mathcal{G}\to[0,1]$ be an ISA. By Gleason's theorem
there exists a unique positive semi-definite self-adjoint operator 
of the trace class $\rho$ such that $\Psi$ is given by the Born rule with
respect to $\rho$.\footnote{As remarked in \cite{WK}: ``Prior to the Bell and Kochen-Specker theorems, Gleason's theorem demonstrated that, for any quantum system of dimension at least three, the unique way to assign probabilities to the outcomes of projective measurements is via the Born rule. In particular, Gleason's theorem excludes any deterministic probability rule given by a \{0, 1\}-valued assignment of probabilities to all the self-adjoint projections on the system's Hilbert space.''}
\item Consider the function $\tau:[0,1]\to\{0,1\}$, 
where $\tau(t)=0$ if and only if $t=0$. Now, given a 
GIV $\Psi:\mathcal{G}\to[0,1]$, the map $\tau \Psi:\mathcal{G}\to\{0,1\}$
is a well-defined map. 
\end{enumerate}
\end{proof}

\begin{definition}
Let $\mathcal{G}$ be a graph. We define a \emph{context} as a complete subgraph (or aggregate) inside $\mathcal{G}$. For example, let $P_1,P_2$ be two elements of $\mathcal{G}$. Then, 
$\{P_1, P_2\}$ is a contexts if $P_1$ is related to $P_2$, $P_1\sim P_2$. Saying it differently, if there exists an edge between $P_1$ and $P_2$. In general, a collection of elements $\{P_i\}_{i\in I}\subseteq \mathcal{G}$ determine a context if $P_i\sim P_j$ for all $i,j\in I$. Equivalently, if the subgraph with nodes $\{P_i\}_{i\in I}$ is complete.  A \emph{maximal} context is a context not contained properly in another context.  If we do not indicate the opposite, when we refer to contexts we will be implying maximal contexts.
\end{definition}

For the graph of powers, the notion of context coincides with the usual one; a complete set of commuting operators. However, all projection operators can be assigned a consistent value bypassing in this way the famous Kochen-Specker theorem,
\begin{theorem} 
{\sc (Intensive Non-Contextuality Theorem)} Given any Hilbert space $H$, then an ISA is possible over $H$.
\end{theorem} 
\begin{proof}
See \cite{deRondeMassri21a}.
\end{proof}

\noindent This theorem restores the possibility of an invariant physical representation of any quantum wave function $\Psi$. Thus, contrary to the orthodox interpretation of QM in terms of systems with properties (which impose a binary valuation), our conceptual representation of quantum physical reality is not relative to any particular context, it is global and essentially intensive. We refer the reader to \cite{deRondeMassri19a, deRondeMassri21a} for a detailed discussion and analysis. Taking this approach as a standpoint, it was possible to derive in \cite{deRondeMassri23}, two important theorems: the Basis Invariance Theorem and the Factorization Invariance Theorem to which we now turn our attention.

\section{An Invariant Account of Bases and Factorizations}

Let us begin by recalling the Basis Invariance Theorem.  
\begin{theorem}{\sc (Basis Invariance Theorem)}
Let $H$ be a Hilbert space, $\mathcal{G}$ its graph of powers and let 
$\Psi:\mathcal{G}\to[0,1]$ be an ISA. Let $\mathcal{C}_1,\mathcal{C}_2\subseteq\mathcal{G}$ be two contexts. The intensities of $\Psi$ over $\mathcal{C}_1$ and $\mathcal{C}_2$ may be different, but both determine in a unique way the whole ISA. In other words, the GIV defined by $\Psi$ does not depend on the basis for $H$. 
\end{theorem}

\begin{proof} Let $\rho$ be the density matrix associated to $\Psi$ in the basis $\mathcal{C}_1$. Then, the intensity of a particular power $P$ can be computed with the Born rule $Tr(\rho P)$. If $\rho'$ is the density matrix associated to $\Psi$ in the basis $\mathcal{C}_2$, then there exists a unitary matrix $U$ such that $\rho'=U\rho U^\dag$. Hence, $P$ is transformed to $P'=U^\dag P U$ and
the intensities of $P$ and $P'$ are the same,
\[
Tr(\rho P) = Tr(U \rho' U^\dag P) = Tr(\rho' U^\dag P U) = Tr(\rho' P').
\]
\end{proof}

\noindent As a consequence of the {\it Basis Invariance Theorem}, the choice of a basis or context becomes in this scheme just the choice of a viewpoint of analysis which is compatible with the choice of any other viewpoint. Thus, escaping contextuality and relativism, all reference frames provide a consistent account ---of the different values of projection operators which are part--- of {\it the same} (intensively) defined state of affairs. 

\smallskip 

Turning now our attention to the relativism of factorizations exposed in \cite{DelaTorre10} it is possible to reach an equivalent result \cite{deRondeMassri23}. Let us denote by ${\cal ISA}(H)$ the set of all possible ISAs on  the Hilbert space $H$. We can identify this set (after fixing a basis) with the space of density matrices over $H$.
In fact, given any completely positive trace-preserving 
map $T:B(H_1)\to B(H_2)$, we obtain a map $T_{*}:{\cal ISA}(H_1)\to {\cal ISA}(H_2)$.
An example of a completely positive trace-preserving map is the partial trace.
Notice that for each factorization $H=H_1\otimes H_2$, we have 
a different partial trace, $T:B(H)\to B(H_1)$.
\begin{definition}
Let $T:B(H)\to B(H')$ be a
completely positive trace-preserving map between (trace-class) operators on Hilbert spaces $H,H'$ and
let $\Psi:\mathcal{G}(H)\to[0,1]$ be an ISA over $H$.
We say that $\Psi':\mathcal{G}(H')\to[0,1]$ is a \emph{shadow}\footnote{It it argued orthodoxly that if $\rho$ is a density matrix in $A \otimes B$, then the matrix associated to the ``sub-system'' in $A$ is the first partial trace of $\rho$. However, the resulting density matrix $Tr_1(\rho)$ is not a submatrix of $\rho$ but a projection, an image. In order to make explicit our departure from the orthodox interpretation in terms of systems and sub-systems which has been already shown to be inconsistent (see \cite{Clifton95, Clifton96}), we choose the word ``shadow'' ---instead of ``subsystem''--- which seems a much more appropriate notion to capture this mathematical aspect of the formalism.}
of $\Psi$ if $\Psi'=T_*(\Psi)$.
Specifically, if $\Psi$ is given by a density matrix $\rho$. 
Then, $\Psi'$ is given by $\rho'=T(\rho)$.
\end{definition}

\begin{theorem} {\sc (Factorization Invariance Theorem)}
Let $H$ be a Hilbert space with factorizations $H=H_1\otimes H_2=
H_1'\otimes H_2'$. Let $T:B(H)\to B(H_1)$ and $T':B(H)\to B(H_1')$
be the partial traces of these factorizations.
Let $\Psi$ be an ISA over $H$ and let $\Psi_1=T_*(\Psi)$ and $\Psi_1'=T_*'(\Psi)$ the corresponding ISAs over $H_1$
and $H_1'$. 
Assume there exists a completely positive trace-preserving map 
$U:B(H_1)\to B(H_1')$
such that $UT=T'$. Then, $\Psi_1'=U_* (\Psi_1)$.
\end{theorem}
\begin{proof}
Straightforward, $\Psi_1'=T'_*(\Psi)=(UT)_*(\Psi)=U_* T_* (\Psi)=U_*(\Psi_1)$. (For an example see \cite{deRondeMassri23})
\end{proof}

The previous theorem implies that all factorizations are consistent and that all corresponding {\it shadows} of the ISA are all compatible between each other. In fact, if we have a compatibility $U$ between the factors $H_1$ and $H_1'$, then this compatibility translates itself in a compatibility between $\Psi_1$ and $\Psi_1'$. 
\begin{corollary}
All factorizations are compatible with respect to the same ISA.
\end{corollary}
\begin{proof}
Follows from the previous Theorem.
\end{proof}

\noindent Thus, the choice of the factorization has no incidence in the intensities or potentia already contained in the considered ISA. All choices of different factorizations remain consistent with the same global account of {\it the same} ISA. Taking these mathematical facts as a standpoint, we are now in conditions to provide a natural conceptual account of both `bases' and `factorizations' operationally linked to the intensive relations between the multiple degrees of freedom that can be considered within a laboratory when analyzing quantum phenomena.

\section{A Conceptual-Operational Account of a Quantum Lab}

Marking a radical departure with respect to the orthodox Bohrian-positivist contemporary account of observations not only as unproblematic {\it givens} of experience but also as the kernel standpoint from which scientific theories can be developed, Einstein writes in the following terms: 
\begin{quotation}
\noindent {\small ``From Hume Kant had learned that there are concepts (as, for example, that of causal connection), which play a dominating role in our thinking, and which, nevertheless, can not be deduced by means of a logical process from the empirically given (a fact which several empiricists recognize, it is true, but seem always again to forget). What justifies the use of such concepts? Suppose he had replied in this sense: Thinking is necessary in order to understand the empirically given, {\it and concepts and `categories' are necessary as indispensable elements of thinking.}'' \cite[p. 678]{Einstein65} (emphasis in the original)}
\end{quotation} 
As it should be recognized, in a truly realist framework, the {\it moments of unity} are not presupposed as exterior (pre-theoretical) givens but ---instead--- are always theoretically constructed in formal-conceptual terms, and that is the reason why ``it is only the theory which decides what can be observed.'' It is only with adequate concepts that we can establish what has been observed. A physical theory requires an invariant formalism, as well as an objective conceptual representation that is consistent with it. Furthermore, both conceptual and formal representations must be operational. This is, they must contain the conditions under which we can understand observations in accordance to what is presented in formal-conceptual terms. We have to be able to consistently connect the theory ---namely, the mathematical-conceptual system--- with the phenomena in question. For as stated clearly by Einstein \cite[p. 26]{Einstein20}, ``The concept does not exist for the physicist until he has the possibility of discovering whether or not it is fulfilled in an actual case.'' If, as what just said, the understanding of observation requires concepts, at the same time, the physical concepts require a link to experience; i.e., the specific operational conditions that allow to observe what the concepts are referring to. So let's advance in this direction. In order to bridge the gap between, on the one hand, the mathematical formalism and the conceptual objectivity as we presented it ---without the projection postulate and replacing the reference to binary outcomes by intensive patterns and values--- and, on the other, a consistent and coherent understanding of experience, we must produce a set of meaningful operational notions. Let us remark that, contrary to the Bohrian doctrine of classical concepts, these notions should not be considered as relying on a presupposed classical representation. Even though we might use the same words, in the context of QM these physical notions need to be understood in internal terms with respect to the theory ---there is no necessary link to the classical representation of physics. As remarked by Heisenberg in an interview by Thomas Kuhn: 
\begin{quotation}
\noindent {\small ``The decisive step is always a rather discontinuous step. You can never hope to go by small steps nearer and nearer to the real theory; at one point you are bound to jump, you must really leave the old concepts and try something new... in any case you can't keep the old concepts.''   \cite[p. 98]{Bokulich06}}
\end{quotation} 
Of course, this does not mean you cannot keep the same words. For example, the words `space' and `time' refer to specific concepts in classical mechanics which differ significantly from the concepts, related to exactly the same words, discussed in relativity theory. This is a consequence of the fact that concepts acquire their meaning in each case from the place they occupy in the specific conceptual system of the particular theory they are part of, this is, their meaning is determined through the relation with the other concepts inside the system. As Heisenberg explained: 
\begin{quotation}
\noindent {\small ``New phenomena that had been observed could only be understood by new concepts which were adapted to the new phenomena. [...] These new concepts again could be connected in a closed system. [...] This problem arose at once when the theory of special relativity had been discovered. The concepts of space and time belonged to both Newtonian mechanics and to the theory of relativity. But space and time in Newtonian mechanics were independent; in the theory of relativity they were connected.'' \cite[pp. 97-98]{Heis58}}
\end{quotation} 

We are now ready to introduce the notions of {\it screen}, {\it detector}, {\it experimental arrangement} and {\it quantum laboratory} that will help us to build a conceptual-operational bridge between the mathematical formalism and quantum phenomena. Let us start with the simplest case. A {\bf screen} with $n$ {\bf detectors} corresponds to the vector space $\mathbb{C}^n$. Choosing a basis, say $\{|1\rangle,\dots,|n\rangle\}$, is the same as choosing some specific $n$ detectors. For example, in the next two figures we picture $\mathbb{C}^2$, and $\mathbb{C}^2$ with a chosen basis,
\begin{center}
\includegraphics[scale=.45]{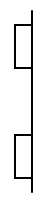}
\captionof{figure}{A screen with a place for two detectors.}
\end{center}
\begin{center}
\includegraphics[scale=.45]{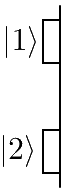}
\captionof{figure}{A screen with two specific detectors.}
\end{center}
A {\bf factorization} $\mathbb{C}^{i_1}\otimes\dots \otimes\mathbb{C}^{i_n}$ is the specific number $n$ of screens, where the screen number $k$ has $i_k$ places for detectors, $k=1,\dots,n$. Choosing a basis in each factor corresponds to choosing the specific detectors; for instance $|\uparrow\rangle, |\downarrow\rangle$. After choosing 
a basis in each factor, we get a basis of the factorization $\mathbb{C}^{i_1}\otimes\dots \otimes\mathbb{C}^{i_n}$
that we denote as
\[
\{ |k_1\dots k_n\rangle \}_{1\le k_j\le i_j}.
\]
The {\bf basis} element $|k_1\dots k_n\rangle$ represents the {\bf power of action} (or simply \textbf{power}) that produces an intensive global effect in the $k_1$ detector of the screen $1$,  in the $k_2$ detector of the screen $2$ and so on until the $k_n$ detector of the screen $n$. In general, any given power will produce a unitary multi-screen non-local effect that has an intensive content.\footnote{It might be stressed that, once we give up Bohr's {\it correspondence principle} or the so called ``quantum to classical limit'', namely, the idea that QM must be related to a spatial continuous representation, we also loose the classical notion of locality which requires a continuous representation precluded by the quantum formalism. This is a natural consequence of Planck's quantum postulate according to which energy must be considered as discrete, and consequently also space and time. In fact, according to the formalism the only distance that we could measure is the distance between basis elements which is always $\sqrt{2}$ (by Pythagoras theorem). This mathematical distance is purely abstract and bears no relation with a distance in Euclidean physical space.} For example, in Figure \ref{det23} we picture two screens, one with two detectors and the other with three detectors. Given the possible combinations, we obtain six different powers $|11\rangle,|12\rangle,|13\rangle,|21\rangle,|22\rangle,|23\rangle$. The powers $|11\rangle$ and $|23\rangle$ are highlighted,
\begin{center}\label{det23}
\includegraphics[scale=.45]{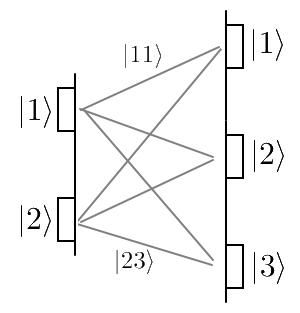}
\captionof{figure}{The 6 {\it powers of action} that arise when considering\\ one screen with 2 detectors and another with 3.}
\end{center}
\noindent According to this representation, a single power of action will produce non-local effects in multiple screens which will be strictly determined in terms of {\it definite valued intensities}. Let us stress the main point of our reference to an ISA, namely, that QM only talks about intensities, never about single outcomes. In particular, an intensity equal to 1 should not be confused with the certainty of obtaining a single outcome; all predictions in QM provide {\it certain knowledge} when considering intensive values ---which is what the theory talks about. This implies a shift from the {\it binary certainty} that was dogmatically imposed by Dirac ---through the {\it ad hoc} addition of the projection postulate--- to an {\it intensive certainty} theoretically imposed by the invariance present in Heisenberg's original quantum formalism itself, where the intensity equal to 1 is just a particular case of the possible values of intensities within the closed interval $[0,1]$. Let us also remark that the operationality we are proposing here allows us to restore a {\it causal} representation ---completely analogous to classical theories--- which not only describes the interaction and evolution of the powers of action in a given situation but also connects them with the intensities obtained in measurement procedures. This methodological move allows us to abandon the Bohrian reference ---established in SQM--- to ``quantum particles'' (represented by quantum states) as well as the {\it non-causal} ``collapse'' evolution imposed by Dirac to single (binary) measurement outcomes. As was pointed out earlier, when we have experiments where we are faced with a single outcome at a time, we must consider these single outcomes not as the expression of particles, and consequently as the main thing to explain, but as a partial information provided within measurement procedures. A single outcome must be then considered as a minimal operational expression of a power. And thus, to fully measure the intensive physical element the theory talks about in these experiments, it is necessary to repeat the experiment until the number of clicks allows us to establish the intensity described by the theory. Besides, the fact the we observe multiscreen effects is a supplementary proof of the fact we are not talking about particles (this effect is incomprehensible if we suppose particles) but about powers of action.

Given an ISA $\Psi$, a factorization $\mathbb{C}^{i_1}\otimes\dots \otimes\mathbb{C}^{i_n}$ and a basis $B=\{|k_1\dots k_n\rangle\langle k_1\dots k_n|\}$ of cardinality $N=i_1\dots i_n$, we define an {\bf experimental arrangement} denoted $\EA_{\Psi,B}^{N;i_1\dots i_n}$, as a specific choice of screens with detectors together with the potentia of each power, that is,
\[
\EA_{\Psi,B}^{N,i_1\dots i_n}= \sum_{k_1,k_1'=1}^{i_1}\dots \sum_{k_n,k_n'=1}^{i_n} 
\alpha_{k_1,\dots,k_n}^{k_1',\dots,k_n'}|k_1\dots k_n\rangle\langle k_1'\dots k_n'|.
\]
The number  that accompanies the power $|k_1\dots k_n\rangle \langle k_1\dots k_n|$ is its {\bf potentia} (or intensity) and the basis $B$ is the set of powers defined by the specific choice of screens with detectors. The number $N$ which is the cardinal of $B$ is called the {\bf degree of complexity} (or degree) of the experimental arrangement.
%begin%%%%%%%%%%%%%
Notice that in the case where we have an experimental arrangement with one screen, then $\EA_{\Psi,B}^{N;n}$ is a rank 1 matrix which can be understood as a vector ---like in Dirac's formulation---, with two screens, $\EA_{\Psi,B}^{N;i.j}$ is a density matrix ---which are orthodoxly interpreted as {\it mixtures}\footnote{See for a detailed analysis \cite{deRondeMassri22b}.}, namely, as the convex sum of pure states--- and in general, for many screens, $\EA_{\Psi,B}^{N;i_1\dots i_n}$ is a (multi)-tensor product. In contraposition to the orthodox account, a product state (or factorization) is not the separation of a  system (defined atomistically) into many parts (sub-atomic particles or sub-systems) but simply the way in which a single power is non-locally expressed on different screens.

%end%%%%%%%%%%%%%%%%
Finally, we use {\bf quantum laboratory}\footnote{We use the term quantum lab in order to stress the fact that a laboratory is a concept that need not be necessarily linked to the substantialist classical representation.} (or quantum lab or Q-Lab)
as a synonim of ISA, in other words, a Q-Lab is the same as an ISA.
%begin%%%%%%%%%%%%%%%%
Now, given an $\EA_{\Psi,B}^{N;i_1\dots i_n}$ 
and some of the screens $k_1,\dots,k_s$,
we define the {\bf shadow} of $\Psi$ in the screens $k_1,\dots,k_s$
as the partial trace in the factors $k_1,\dots,k_s$ of $\EA_{\Psi,B}^{N;i_1\dots i_n}$,
\[
\EA_{\Psi,B'}^{K;k_1\dots k_s} := Tr_{k_1,\dots,k_s}(\EA_{\Psi,B}^{N;i_1\dots i_n}),
\quad K = k_1\dots k_s.
\]
%end%%%%%%%%%%%%%%%%

To sum up, a Q-Lab is denoted as $\Psi$ and is defined as a graph of powers (all the powers) together with  its intensities (see \cite{deRondeMassri21a} for the precise definition) containing all possible experimental set ups. An equivalent class of experimental arrangements produced within the same lab $\EA_{\Psi}^{N}$ consists of all the specific experimental arrangements produced by choosing different configurations of screens and detectors such that the total number of powers is $N$; i.e.,  $\EA_{\Psi, B}^{N}, \EA_{\Psi, B'}^{N}, \EA_{\Psi, B''}^{N}$. Notice that the factorization is already included in the information contained within each basis. Let us notice that our notation allows to include the information of the particular arrangement, namely, the number of screens and detectors that are placed in each of the screens. Let us also remark that, as explained in \cite{deRondeMassri21a}, choosing an  $\EA_{\Psi, B}^{N}$ is the same as choosing a specific context within the graph of powers:
\begin{center}
\includegraphics[scale=.4]{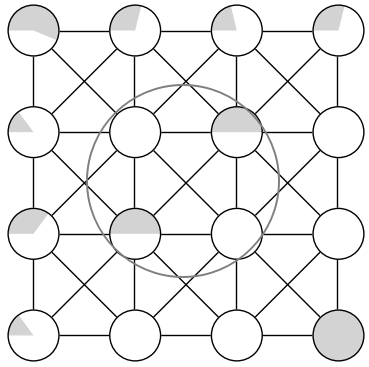}
\captionof{figure}{An $\EA_{\Psi,B}^{4}$ chosen from \\ the net of powers given by the Q-Lab $\Psi$.}
\end{center}

We are now in conditions to address more in detail the essential equivalence relations we have reached for QM through the formal account of bases and factorizations as well as their conceptual representation.

\section{Equivalence Relations in Quantum Mechanics}

While in section 4 we presented the basis and factorization invariance theorems, in section 5 we introduced a series of operational physical notions capable to express the mathematical formalism of the theory. We are now in conditions to think of both theorems, in truly physical terms, as {\it equivalence relations} of the theory, something essential for QM, understood in physical terms, as providing a relation between {\it different} reference frames when discussing about {\it the same} state of affairs. These relations might be considered as analogous to Poincare's {\it relativity principle}\footnote{As stated by Einstein \cite[p. 98]{Einstein20}: ``{\bf Special principle of relativity:} If a system of coordinates {\it K} is chosen so that, in relation to it, physical laws hold good in their simplest form, the same laws hold good in relation to any other system of coordinates {\it K'} moving in uniform translation relatively to {\it K}.''} which guided Einstein in his search to develop the Special Theory of Relativity beyond the ontological picture ---transformed in ``common sense'' within modernity--- of physical reality in terms of entities within absolute space and time.\footnote{It is important to emphasize that this reference to ``relativity'' cannot be equated with the {\it perspectival relativism} that can be found in Dirac's re-formulation of the notion of (quantum) state in which the different basis dependent representations cannot be globally considered.} Einstein stressed the need to hold fast to the equivalence relation between different reference frames ---as stated by Poincare's principle--- as a necessary condition for a physical theory, and to the empirical finding of the invariance of the speed of light, even if this meant doing away with one of the fundamental aspects of the classical representation and replacing the Galilean transformations ---grounded on our ``commonsensical'' classical Newtonian spatiotemporal representation--- by the newly developed Lorentz transformations. In our case we also hold fast to the need of providing an objective-invariant relation between the different basis-dependent representations of QM, which is clearly possible, but it entails rejecting the atomist worldview of the classical representation. Given that physics is the search for unity within change, the invariant-objective relations between {\it different} bases and factorizations which allow to describe {\it the same} becomes essential. Indeed, physics attempts to define {\it moments of unity} not only in order to consider a multiplicity of different phenomena in a consistent manner, but also to account for the equivalence between representations, dependent on empirical and formal perspectives (i.e., reference frames). There must always exist an equivalence relation which consistently connects the (different) possible descriptions provided by (different) reference frames. This is, in fact, the condition which allowed Galileo and Newton to unite the apparently distinct celestial and terrestrial phenomena. This is called in physics {\it the principle of relativity}, namely, the requirement that the equations describing the order of physical relations have the same form in different yet equivalent frames of reference. The relativity principle has played an essential role within physical theories, from classical mechanics and electromagnetism up to relativity theory.\footnote{One decade later, Einstein would then extend it to the {\it general principle of relativity} which states that all systems of reference are equivalent with respect to the formulation of the fundamental laws of physics \cite[p. 220]{Moller52}.} The reason is obvious: without such a principle it becomes simply impossible to talk about {\it the same} experience in different yet equivalent situations, and consequently, the production of a meaningful physical discourse about the possible experiences in different labs becomes precluded. Without such a principle it becomes simply impossible to say we are performing {\it the same} experiment in different labs and physical theories become operationally impracticable. 

Going back to SQM, it becomes then clear the radical subversion of the discipline that took place in the 20th century through Bohr and Dirac's re-definition of the notion of (quantum) {\it state} in (complementary) basis dependent terms (section 2). As we discussed above, the restriction to binary values (`clicks' in detectors) involves the abolition of operational-invariance; i.e., the consistent translation between the experience found in different reference frames of {\it the same} state of affairs. The idea that the state of a system is dependent on a particular basis (i.e., a reference frame) breaks down the very meaning of the notion of {\it state} as referring to something which remains the same independently of reference frames. For if a state is different in different bases it becomes nothing but different to itself, a self-contradiction. This is, of course, the disintegration of the possibility of considering any meaningful {\it moment of unity}, any common reference within the theory. Of course, for many, this has not been understood as ``the end of physics'', but rather, as the production of a new (anti-realist) ``instrumentalist'' account of the discipline where the reference to the world and reality was replaced by the ``more productive'' accurate prediction of measurement outcomes. As described by Arthur Fine: 
\begin{quotation}
\noindent{\small ``[The] instrumentalist moves, away from a realist construal of the emerging quantum theory, were given particular force by Bohr's so-called `philosophy of complementarity'; and this nonrealist position was consolidated at the time of the famous Solvay conference, in October of 1927, and is firmly in place today. Such quantum nonrealism is part of what every graduate physicist learns and practices. It is the conceptual backdrop to all the brilliant success in atomic, nuclear, and particle physics over the past fifty years. Physicists have learned to think about their theory in a highly nonrealist way, and doing just that has brought about the most marvelous predictive success in the history of science.''  \cite[p. 1195]{PS}}
\end{quotation}
Even though we agree with Fine's historical account of the transformation of physics, we strongly disagree with his conclusions. It is certainly true that physicists learned through the work of Bohr to think about theories in a highly nonrealist way, but this new praxis did not produce ``the most marvelous predictive success in the history of science'', but ---on the very contrary--- the destruction of the guiding principles of physics as well as the the prohibition of essential lines of research. A very good example is the banning of Einstein's critical work in QM which led to the prohibition of investigating quantum entanglement, a notion which remained erased from physics for almost half a century. On the contrary, going back to the realist program of science, the logos approach re-connects with the search for theoretical (formal-conceptual) unity within experience through the guidance of the principles of operational-invariance and objectivity. Thus, escaping complementarity and the relativist account of states, objects and outcomes, we will now consider the basis and factorization invariance theorems (section 4) as general equivalence relations of the theory of quanta. 

\smallskip 

Given two ISAs $\Psi_1$, $\Psi_2$ and two bases $B_1$ and $B_2$ of the same space we can define four different EAs, $\EA_{\Psi_1,B_1}^{N;i_1\dots i_s},
\EA_{\Psi_1,B_2}^{N;i_1\dots i_s},\EA_{\Psi_2,B_1}^{N;i_1\dots i_s}$ and $\EA_{\Psi_2,B_2}^{N;i_1\dots i_s}$. Notice that all four EAs are different, but in different ways. The EAs $\EA_{\Psi_1,B_1}^{N;i_1\dots i_s}$ and $\EA_{\Psi_1,B_2}^{N;i_1\dots i_s}$ are different but can be converted from one to another by a change of basis (rearranging the screens and detectors). This means they can be understood as making reference to the same (intensive) state of affairs. The same counts for $\EA_{\Psi_2,B_1}^{N;i_1\dots i_s}$ and $\EA_{\Psi_2,B_2}^{N;i_1\dots i_s}$. However, the experimental arrangements $\EA_{\Psi_1,B_1}^{N;i_1\dots i_s}$  and $\EA_{\Psi_2,B_1}^{N;i_1\dots i_s}$ even though have the same basis, are related to different states, this means there is no obvious way of converting one to the other through a transformation (same for $\EA_{\Psi_1,B_2}^{N;i_1\dots i_s}$  and $\EA_{\Psi_2,B_2}^{N;i_1\dots i_s}$).
\[
\xymatrix{
\EA_{\Psi_1,B_1}^{N;i_1\dots i_s}\ar@{}[r]|\equiv\ar@{}|{\not\equiv}[d]&\EA_{\Psi_1,B_2}^{N;i_1\dots i_s}\ar@{}|{\not\equiv}[d]\\
\EA_{\Psi_2,B_1}^{N;i_1\dots i_s}\ar@{}[r]|\equiv&\EA_{\Psi_2,B_2}^{N;i_1\dots i_s}
}
\]

\smallskip

Now, in the light of this new conceptualization of the mathematical formalism, we can rephrase the basis and factorization invariance theorems in the following intuitive terms:
\begin{theorem}{\sc (Basis Invariance Theorem)}
Given a specific quantum laboratory, all experimental arrangements of the same degree are equivalent. 
\end{theorem}

\begin{theorem} {\sc (Factorization Invariance Theorem)}
The experiments performed with an 
experimental arrangement of degree $N$
can also be performed with an  experimental arrangement
of degree $N+M$.
\end{theorem}
\noindent While the {\it Basis Invariance Theorem} implies that the knowledge we obtain from  the intensities in one experimental arrangement is equivalent to the knowledge obtainable in any other experimental arrangement of the same degree, the {\it Factorization Invariance Theorem} tells us that reducing the complexity of an experimental arrangement  will not increase our knowledge in any way.
%begin%%%%%%%%%%%%%%%%%%
In the light of these Theorems, let us make some comments 
regarding the notation $\EA_{\Psi,B}^{N;i_1\dots i_n}$.
First of all, choosing a different basis $B'$ in the same factorization,
by the Basis Invariance Theorem, we obtain an equivalent
experimental arrangement, $\EA_{\Psi,B'}^{N;i_1\dots i_n}$,
\[
\EA_{\Psi,B}^{N;i_1\dots i_n}\equiv \EA_{\Psi,B'}^{N;i_1\dots i_n}.
\]
Hence, we denote $\EA_{\Psi}^{N;i_1\dots i_n}$ to an experimental arrangement with
a fixed number of screens and detectors but without specifying which detectors are chosen.
Another possibility regarding the notation is to choose a different factorization
of the same vector space $N=j_1\dots j_{n'}$ (and consequently another set of powers $B''$). Then, in this situation we also have by the Basis Invariance Theorem
\[
\EA_{\Psi,B}^{N;i_1\dots i_n}\equiv \EA_{\Psi,B''}^{N;j_1\dots j_{n'}}.
\]
Hence, we denote $\EA_{\Psi}^N$ to an experimental arrangement of degree $N$
without specifying which screens and which detectors are chosen. The last 
possibility to consider comes from the relation between experimental arrangements
of different degrees of complexity. 
By the Factorization Invariance Theorem, we have the following relation
\[
\EA_{\Psi}^N\Rrightarrow\EA_{\Psi}^{N+M}.
\]
Hence, we denote $\EA_{\Psi}$ to an experimental arrangement of any degree of complexity.
%end%%%%%%%%%%%%%%%%%%%
\begin{remark}
Given a factorization there are multiple bases that can be chosen in each factor. However, a basis of the whole space determines a specific factorization uniquely. 
\[
\text{Factorization}\not\Rightarrow\text{Basis}\]
\[
\text{Basis}\Rightarrow\text{Factorization}\quad 
\]
In conceptual terms, given an $\EA_{\Psi}^{N}$ of some quantum laboratory $\Psi$ 
there are different possible experimental arrangements that can be constructed but none which is uniquely determined. However, given an $\EA_{\Psi,B}^{N;i_1\dots i_n}$ it is possible to determine all other experimental arrangements $\EA_{\Psi}^{N}$ that can be constructed.
\[
\EA_{\Psi}^{N}\not\Rightarrow\EA_{\Psi,B}^{N;i_1\dots i_n}\]
\[
\EA_{\Psi,B}^{N;i_1\dots i_n}\Rightarrow\EA_{\Psi}^{N}\quad 
\]
\end{remark}

\begin{corollary}
Given a Quantum Lab, all experimental arrangements are Globally Consistent.
\end{corollary}
\begin{proof}
Follows from the previous theorems.
\end{proof}

We reach in QM a new form of equivalence determined by the role of bases as reference frames. While in the classical case a reference frame must be understood as a formal perspective providing the specific values of all the properties which determine a given state of affairs and these values can be translated via Galilean transformation to any other reference frame, what is truly new in the quantum case is that while a basis provides the intensive values only of a specific set of powers the transformation to a new basis allows to find out the specific intensive values of a different set of powers of the same (intensive) state of affairs.

\section*{Conclusions}

In this paper we provided a consistent and coherent objective-invariant account of bases and factorizations through the reference to an intensive ---rather than binary--- state of affairs. We did so through the addition of new physical concepts such as power of action, quantum laboratory, experimental arrangement, etc. With the aid of these concepts we provided a new, more intuitive, account of the basis and factorization invariance theorems. Furthermore, we provided a deeper understanding of the complex relation between factorizations and bases and between a quantum lab and the different experimental arrangements that can be built within it.

\section*{Acknowledgements} 

This work was partially supported by the following grants:  ANID-FONDECYT, Project number: 3240436.

\bibliographystyle{plain}

\end{document}